\documentclass[aps,pre,preprint,floatfix]{revtex4}

\usepackage{latexsym}
\usepackage[usenames]{color}
\usepackage{amsthm}
\usepackage{hyperref}
\usepackage{braket}

\usepackage{amsmath}    
\usepackage{amsfonts}  
\usepackage{amssymb}
\usepackage{graphicx}
\usepackage{slashed}
\usepackage{fouridx}
\usepackage{tikz}
\usetikzlibrary{calc}

\theoremstyle{plain} 
\theoremstyle{plain} 
\theoremstyle{plain} \newtheorem{theor}{Theorem}[section] 
\theoremstyle{plain} 
\theoremstyle{plain} \newtheorem*{corol}{Corollary} 
\theoremstyle{remark} 
\theoremstyle{plain} 
\theoremstyle{remark}

\newcommand\CROSS[1]{%
  \hbox{%
    \vbox{
      \hrule
      \kern2.5pt
      \hbox{$#1$\,\strut}
    }%
  \vrule
  }\mskip\thickmuskip
}

\tikzset{
solid node/.style={circle,draw,inner sep=1.5,fill=black},
hollow node/.style={circle,draw,inner sep=1.5}
}

\newlength{\arrowsize}  
\pgfarrowsdeclare{biggertip}{biggertip}{  
  \setlength{\arrowsize}{0.5pt}  
  \addtolength{\arrowsize}{.5\pgflinewidth}  
  \pgfarrowsrightextend{0}  
  \pgfarrowsleftextend{-5\arrowsize}  
}{  
  \setlength{\arrowsize}{0.5pt}  
  \addtolength{\arrowsize}{.5\pgflinewidth}  
  \pgfpathmoveto{\pgfpoint{-5\arrowsize}{4\arrowsize}}  
  \pgfpathlineto{\pgfpointorigin}  
  \pgfpathlineto{\pgfpoint{-5\arrowsize}{-4\arrowsize}}  
  \pgfusepathqstroke  
}

\begin{document}
\begin{center}
\Large{\textbf{A theorem on integration \\ based on the digital expansion}}\\ 
~\\
\large{Vladimir Garc\'{\i}a-Morales, Javier Cervera and Jos\'e A. Manzanares}\\

Departament de F\'{\i}sica de la Terra i Termodin\`amica\\ Universitat de Val\`encia, \\ E-46100 Burjassot, Spain
\\ 
\end{center}
\small{}
\noindent  \normalsize{}
~\\ The binary radix expansion of a real number can be used to code the outcome of any series of coin tosses, a fact that provides an intriguing link between number theory, measure theory and statistical physics. Inspired by this fact,  a general result is established for the definite integral of a differentiable function of a single variable that allows any such integral to be exactly written in terms of a double series. The theorem can be directly applied to a wide variety of integrals of physical interest and to derive new series expansions of real numbers and real-valued functions. We apply the theorem to the integration of the equation of motion in one dimension of classical Hamiltonian systems, focusing in the analysis of the nonlinear pendulum.
\pagebreak

\section{Introduction}

In his classical memoir of 1909 \cite{Borel}, \'Emile Borel made the observation that the binary digits of a real number are independent, marking the beginning of modern probability theory \cite{Kacker}.  An important theorem by Borel states that, asymptotically, almost every number has the same number of zeros and ones in its binary expansion. The binary expansion of a real number in the unit interval may be used to represent any sequence of random coin tosses. Indeed, coin tossing provides an elementary example of the way in which probabilistically natural considerations uncover interesting insights into real analysis \cite{Stroock}. A number with $n$ binary digits represents a sequence of $n$ independent coin tosses and the probability of obtaining that specific sequence is $1/2^n$ \cite{Kacker}.  

This connection between binary digits of numbers and coin tosses provides thus a straightforward approach to the notion of statistical independence. The latter was exploited by Mark Kac in his book \cite{Kacker} investigating the connection with measure theory and  deriving intriguing consequences when integration is involved. He showed how certain integrals involving $n$ independent variables can be factorized in the product of $n$ independent integrals. In this way, Kac not only found the proof of Vieta's formula for $\pi$ that opens his book but also explained how the digital expansion is related to the law of large numbers and the central limit theorem. Kac's approach can be extended to obtain deep results in the theory of normal numbers \cite{Goodman, Mendes}. Another work has elucidated and emphasized the elementary character of the Kac and Goodman approaches to normal numbers \cite{Nielsen}.

The Rademacher functions are defined in terms of the digits of the binary expansion of a real number. The connection of these functions with probability theory and the intuitions behind Kac's approach may bridge a gap between deterministic physical models and their statistical description, as shown in a recent work where local Poisson and local central limit theorems are proved for a simple deterministic ``Bernoulli'' model \cite{Beck}.

The connection between probability and integration has enormous practical use, as illustrated by e.g. Monte Carlo integration techniques which serve to numerically compute any definite integral with help of random numbers. Since the binary expansion is connected to probability and measure theory, we can ask whether binary expansions may provide a systematic method to calculate integrals. In this article, we give an affirmative answer to this question. We prove a theorem that yields a rather general means to compute definite integrals through series summation. The theorem directly applies to any differentiable function $f(x)$. In contrast with the Monte Carlo method of integration, which is essentially numerical, the theorem here can be exploited both numerically and analytically (in those cases where the double series can be explicitly handled).

The outline of this article is as follows. In Section \ref{digital} we introduce the digital expansion \cite{CHAOSOLFRAC, replace}. In Section \ref{main}, the theorem on integration and a number of corollaries are proved using this expansion. In Section \ref{discussion} some examples of application are discussed. Finally, in Section \ref{Hamiltonian} the integration of the equation of motion of Hamiltonian systems in one dimension is considered as a specific example in mathematical physics. We consider the libration motion of a nonlinear pendulum \cite{Chernikov} and write the resulting elliptic integral in terms of a double series that captures, to arbitrary precision, the value of the period.

\section{The digital expansion} \label{digital}

A real number $x \ge 0$ can be expanded in radix (base) $p \in \mathbb{N}$ as
\begin{equation}
x= a_{N}p^{N}+a_{N-1}p^{N-1}+\ldots+a_{1}p+a_{0}+a_{-1}p^{-1}+a_{2}p^{-2}+\ldots \label{bexpareal2}
\end{equation}
where the digits $a_{k} \in [0,p-1]$ are non-negative integers and $N= \lfloor \log_{p} x \rfloor$. 
By subtraction of 
\begin{equation}
\left \lfloor \frac{x}{p^{k}} \right \rfloor=a_{N}p^{N-k}+a_{N-1}p^{N-k-1}+\ldots +a_{k+1}p+a_{k}   \label{eq4}
\end{equation}
and
\begin{equation}
p\left \lfloor \frac{x}{p^{k+1}} \right \rfloor=a_{N}p^{N-k}+a_{N-1}p^{N-k-1}+\ldots +a_{k+1}p   ,   \label{eq5}
\end{equation}
both obtained from Eq. (\ref{bexpareal2}), it is apparent that the digit function
\begin{equation}
\mathbf{d}_{p}(k,x):= \left \lfloor \frac{x}{p^{k}} \right \rfloor-p\left \lfloor \frac{x}{p^{k+1}} \right \rfloor    \label{cucuAreal}
\end{equation}
coincides with the $k$-th digit of $x\ge 0$ when expanded in radix $p$, that is, $a_{k}=\mathbf{d}_{p}(k,x)$. All digits accompanying powers of the radix with exponent higher than $\lfloor \log_p x \rfloor$ are zero,
\begin{equation}
\mathbf{d}_{p}(k,x)=0 \qquad k > \lfloor \log_p x \rfloor . \label{zeroab}
\end{equation}
From Eq. (\ref{cucuAreal}), we also note the scaling relationship for $k, s \in \mathbb{Z}$
\begin{equation}
\mathbf{d}_{p}(k-s,x)=\mathbf{d}_{p}(k,p^{s}x) .  \label{scal}
\end{equation}

For any $x\in \mathbb{R}$, the convergent series in Eq.  (\ref{bexpareal2}) can then be presented as \cite{CHAOSOLFRAC, replace}
\begin{equation}
x=\text{sgn}(x)\sum_{k=-\infty}^{\lfloor \log_{p}|x| \rfloor} p^{k} \mathbf{d}_{p}(k,|x|). \label{idenreal}
\end{equation} 
Similarly, a complex number $z$ can be expanded as \cite{CHAOSOLFRAC}
\begin{equation}
z=\frac{z}{|z|}|z|=\mathrm{e}^{\mathrm{i}\theta}\sum_{k=-\infty}^{\lfloor \log_{p}|z| \rfloor} p^{k} \mathbf{d}_{p}(k,|z|) \label{idencom}
\end{equation}
where  $\rm{e}^{\rm{i}\theta}$ $\equiv z/|z|$ is the phase factor.

\section{Statement and proof of the main result} \label{main}

\begin{theor} \label{digitinte}  Let $f(x)$, $f:[a,b]\to \mathbb{R}$ be a differentiable function in the interval $[a,b]$ ($0\le a\le b$). Then,
\begin{eqnarray}
\int_{a}^{b}f(x)\mathrm{d}x&=&\sum_{k=-\lfloor \log_2 b \rfloor}^{\infty}\frac{1}{2^{k}}\sum_{n=  \left \lfloor 2^{k}a \right \rfloor+1}^{\left \lfloor 2^{k}b \right \rfloor}(-1)^{n+1}f\left(\frac{n}{2^{k}} \right) \label{ectheo}
\end{eqnarray}
where $\left \lfloor \ldots \right \rfloor$ denotes the floor (lower closest integer)  function. Furthermore, if  $f(x)$ has differentiable inverse $f^{-1}(y)$ and $f(b)>f(a)$ then, regardless of whether $a\le b$ or $b\le a$,
\begin{eqnarray}
\int_{a}^{b}f(x)\mathrm{d}x&=&bf(b)-af(a)+\sum_{k=-\lfloor \log_2 f(b) \rfloor}^{\infty}\frac{1}{2^{k}}\sum_{n= \left \lfloor 2^{k}f(a) \right \rfloor+1}^{\left \lfloor 2^{k}f(b) \right \rfloor}(-1)^{n}f^{-1}\left(\frac{n}{2^{k}} \right) .  
\label{ectheo3}
\end{eqnarray}
\end{theor}

\begin{proof} Let $f:[a,b]\to \mathbb{R}$ be a differentiable function. By integration by parts,
\begin{equation}
\int_{a}^{b}f(x)\mathrm{d}x=\int_{a}^{b}\mathrm{d}(xf)-\int_{a}^{b}x\mathrm{d}f(x)=bf(b)-af(a)-\int_{a}^{b}x\mathrm{d}f(x) . \label{vamos0}
\end{equation}
If $f(x)$ has well defined inverse $f^{-1}(y)$, then
\begin{equation}
\int_{a}^{b}x\mathrm{d}f(x)=\int_{a}^{b}f^{-1}(f(x))\mathrm{d}f(x) = \int_{f(a)}^{f(b)}f^{-1}(y)\mathrm{d}y.
\label{vamosend}
\end{equation}
Eq. (\ref{ectheo3}) is obtained from Eq. (\ref{vamosend}) by applying Eq. (\ref{ectheo}) to the integral 
$\int_{f(a)}^{f(b)}f^{-1}(y)\mathrm{d}y$.

Expansion of $x$ in the binary radix and exchange of summation and integration because the series is convergent $\forall x$, lead to
\begin{equation}
\int_{a}^{b}x \mathrm{d}f(x)=\sum_{k=-\infty}^{\lfloor \log_2 b \rfloor}2^{k} \int_{a}^{b}\mathbf{d}_{2}(k,x)\mathrm{d}f(x) , \label{p1}
\end{equation}
where, in writing the upper bound $\lfloor \log_2 b \rfloor$, we have used Eq. (\ref{zeroab}) and took into account that  $x \le b$ $ \forall x \in [a,b]$ when $b \ge a$. If we use of the scaling relationship, Eq. (\ref{scal}), and take $y=2^{-k}x$ and $x=2^{k}y$, we get 
\begin{equation}
\int_{a}^{b}x \mathrm{d}f(x)=\sum_{k=-\infty}^{\lfloor \log_2 b \rfloor}2^{k}\int_{a}^{b}\mathbf{d}_{2}\left(0,2^{-k}x\right)\mathrm{d}f(x)=\sum_{k=-\infty}^{\lfloor \log_2 b \rfloor}2^{k}\int_{2^{-k}a}^{2^{-k}b}\mathbf{d}_{2}(0,y)  \mathrm{d}f(2^k y) . \label{eq15}
\end{equation}
The function $\mathbf{d}_{2}(0,y)$ is non-zero (and equal to one) only if $1\le y-2n<2$, i.e. if $2n+1\le y < 2n+2$, $n\in \mathbb{Z}$ \cite{CHAOSOLFRAC}. Then, since $2^{-k}a \le 2n+1$ and $2n+2 \le 2^{-k}b$, we find $\left  \lceil 2^{-k}a \right \rceil \le 2n+1$ and $2n+2 \le \left \lfloor 2^{-k}b \right \rfloor$ \cite{Knuth}. Therefore, the relevant integers $n$ are those that satisfy $N_{k,-}\le n \le N_{k,+}$ where 
\begin{eqnarray}
N_{k,-}&=&\left \lfloor 2^{-k}a \right \rfloor+1+\mathbf{d}_{2}\left(0, \left \lfloor 2^{-k}a \right \rfloor \right)=2\left(\left \lfloor 2^{-k}a \right \rfloor-\left \lfloor 2^{-k-1}a \right \rfloor           \right)+1  \nonumber \\
N_{k,+}&=&\left \lfloor 2^{-k}b \right \rfloor-\mathbf{d}_{2}\left(0, \left \lfloor 2^{-k}b \right \rfloor \right)=2\left \lfloor 2^{-k-1}b \right \rfloor .
\end{eqnarray}
Note that $N_{k,-}$ is odd and $N_{k,+}$ is even. By splitting the last integral in Eq. (\ref{eq15}) in different contributions
\begin{eqnarray}
&&\sum_{k=-\infty}^{\lfloor \log_2 b \rfloor}2^{k}\int_{2^{-k}a}^{2^{-k}b}\mathbf{d}_{2}(0,y)\mathrm{d}f(2^k y) =\sum_{k=-\infty}^{\lfloor \log_2 b \rfloor}2^{k}\left[\mathbf{d}_{2}\left(0, \left \lfloor 2^{-k}a \right \rfloor \right)\int_{2^{-k}a}^{\left \lfloor 2^{-k}a \right \rfloor+1}\mathrm{d}f(2^k y)\right. \nonumber \\
&&\qquad \qquad \qquad \qquad \qquad \qquad \quad \left.+\sum_{m=0}^{(N_{k,+}-N_{k,-}-1)/2}\int_{N_{k,-}+2m}^{N_{k,-}+2m+1} \mathrm{d}f(2^k y)\right. \nonumber \\
&&\qquad \qquad \qquad \qquad \qquad \qquad \quad \left.+\mathbf{d}_{2}\left(0, \left \lfloor 2^{-k}b \right \rfloor \right)\int_{\left \lfloor 2^{-k}b \right \rfloor}^{2^{k}b}\mathrm{d}f(2^k y)    \right] 
\end{eqnarray}
we obtain
\begin{eqnarray}
\int_{a}^{b}x\mathrm{d}f(x)&=&\sum_{k=-\infty}^{\lfloor \log_2 b \rfloor}2^{k} \left[\mathbf{d}_{2}\left(0, \left \lfloor 2^{-k}b \right \rfloor \right) f(b)-\mathbf{d}_{2}\left(0, \left \lfloor 2^{-k}a \right \rfloor \right) f(a)\right] \nonumber \\
&&+\sum_{k=-\infty}^{\lfloor \log_2 b \rfloor}2^{k}\sum_{n= \left \lfloor 2^{-k}a \right \rfloor+1}^{\left \lfloor 2^{-k}b \right \rfloor}(-1)^{n}f\left(2^{k} n\right)   \nonumber \\
&&=b f(b)-af(a) +\sum_{k=-\lfloor \log_2 b \rfloor}^{\infty}2^{-k}\sum_{n= \left \lfloor 2^{k}a \right \rfloor+1}^{\left \lfloor 2^{k}b \right \rfloor}(-1)^{n}f\left(2^{-k} n\right) ,  \label{vamos}
\end{eqnarray}
where we have used that
\begin{equation}
\sum_{k=-\infty}^{\lfloor \log_2 b \rfloor}2^{k}\mathbf{d}_{2}\left(0, \left \lfloor 2^{-k}x \right \rfloor \right)=\sum_{k=-\infty}^{\lfloor \log_2 b \rfloor}2^{k}\mathbf{d}_{2}\left(0, 2^{-k}x  \right)=\sum_{k=-\infty}^{\lfloor \log_2 b \rfloor}2^{k}\mathbf{d}_{2}\left(k, x  \right)=x .
\end{equation}
Equation (\ref{ectheo}) is obtained by replacing Eq. (\ref{vamos}) in Eq. (\ref{vamos0}). This completes the proof of the theorem. \end{proof}

The requirement that $a\ge 0$ for the lower integration limit can be lifted. 

\begin{corol} \label{coro1}
For $f(x)$ differentiable and $b\ge a$,
\begin{eqnarray}
\int_{a}^{b}f(x)\mathrm{d}x&=&\sum_{k=-\lfloor \log_2 (b-a) \rfloor}^{\infty}\frac{1}{2^{k}}\sum_{n=1}^{\left \lfloor 2^{k}(b-a) \right \rfloor}(-1)^{n+1}f\left(a+\frac{n}{2^{k}} \right) . \label{coro1e}
\end{eqnarray}
\end{corol}

\begin{proof}
The result follows by applying the theorem to the integral
\begin{equation}
\int_{a}^{b}f(x)\mathrm{d}x=\int_{0}^{b-a}f(a+x)\mathrm{d}x .
\end{equation}
\end{proof}

If we truncate the sum over $k$ Eq. (\ref{coro1e}) to $P$ terms, the function
\begin{equation}
I(x,P):=\sum_{k=-\lfloor \log_2 (b-a) \rfloor}^{-\lfloor \log_2 (b-a) \rfloor+P}\frac{1}{2^{k}}\sum_{n=1}^{\left \lfloor 2^{k}(b-a) \right \rfloor}(-1)^{n+1}f\left(a+\frac{n}{2^{k}} \right)  \label{Itrunc}
\end{equation}
approximates the integral numerically. There are $\left \lfloor (b-a)2^{-\lfloor \log_2 (b-a) \rfloor+P} \right \rfloor$ dyadic rationals involved in the tiniest scale (the different scales being governed by the index $k$) and the sum over $n$ runs over finite subsets of these. Therefore, a rough upper bound for the error in numerically aproximating the integral by Eq. (\ref{Itrunc}) is provided by the rectangle method \cite{Apostol} as 
\begin{equation}
\left |I(x,P)-\int_{a}^{b}f(x)\mathrm{d}x \right| < \frac{M_1 (b-a)^2}{2\left \lfloor (b-a)2^{-\lfloor \log_2 (b-a) \rfloor+P} \right \rfloor}
\end{equation}
where $M_1$ is the maximum value of $|f'(x)|$ on the interval.

If $f(x)$ is twice differentiable we can expand $f(x+h)$ in terms of $f(x)$ in a series that is different to the Taylor series since higher derivatives of $f(x)$ are not needed. 

\begin{corol} If $f(x)$ is twice differentiable in an interval of radius $h$ 
\begin{eqnarray}
f(x+h)&=&f(x)+\sum_{k=-\lfloor \log_2 h \rfloor}^{\infty}\frac{1}{2^{k}}\sum_{n=1}^{\left \lfloor 2^{k}h \right \rfloor}(-1)^{n+1}f'\left(x+\frac{n}{2^{k}} \right) . \label{coro2e}
\end{eqnarray}
\end{corol}

\begin{proof} From Eq. (\ref{coro1e}) we have
 \begin{eqnarray}
\int_{0}^{h}f'(x+t)\mathrm{d}t&=&\sum_{k=-\lfloor \log_2 h \rfloor}^{\infty}\frac{1}{2^{k}}\sum_{n=1}^{\left \lfloor 2^{k}h \right \rfloor}(-1)^{n+1}f'\left(x+\frac{n}{2^{k}} \right)=f(x+h)-f(x), \label{ectheo2cor2}
\end{eqnarray}
from which the result follows.
\end{proof}

In the particular case $h=1$, Eq. (\ref{coro2e}) reduces to
\begin{eqnarray}
f(x+1)&=&f(x)+\sum_{k=0}^{\infty}\frac{1}{2^{k}}\sum_{n=1}^{2^{k}}(-1)^{n+1}f'\left(x+\frac{n}{2^{k}} \right) . \label{coro2eSEC}
\end{eqnarray}

\begin{corol} If $f(x)$ is twice differentiable and periodic with period $T$ then 
\begin{eqnarray}
0&=&\sum_{k=-\lfloor \log_2 T \rfloor}^{\infty}\frac{1}{2^{k}}\sum_{n=1}^{\left \lfloor 2^{k}T \right \rfloor}(-1)^{n+1}f'\left(x+\frac{n}{2^{k}} \right) \label{coro3e}
\end{eqnarray}
\end{corol}

\begin{corol} Let $f(x,y)$ be Riemann integrable in $x\in [a,b]$, $y\in [c,d]$, $b>a$, $d>c$, $a,b,c,d\in \mathbb{R}$ and let $g(x)=\int_{c}^{d}f(x,y)dy$ exist for every $x\in [a,b]$. Then,
\begin{eqnarray}
\int_{a}^{b}\int_{c}^{d}f(x,y)\mathrm{d}x\mathrm{d}y&=&\sum_{k=-\lfloor \log_2 b \rfloor}^{\infty}\sum_{h=-\lfloor \log_2 d \rfloor}^{\infty}\sum_{n=\left \lfloor 2^{k}a \right \rfloor+1}^{\left \lfloor 2^{k}b \right \rfloor}
\sum_{n=\left \lfloor 2^{k}c \right \rfloor+1}^{\left \lfloor 2^{k}d \right \rfloor}
\frac{(-1)^{n+m}}{2^{k+h}}f\left(\frac{n}{2^{k}},\frac{m}{2^{h}} \right) \label{ectheo2}
\end{eqnarray}
\end{corol}

\begin{proof} The result follows from applying the theorem twice. Integration and summation can be exchanged because the function $f(x,y)$ satisfies all necessary conditions for the Fubini theorem \cite{DiBene} to be applicable.
\end{proof}


\section{Examples and applications in scale analysis and number theory} \label{discussion}

Our theorem yields an expression that is reminiscent of wavelet theory. The double series in which the integral is transformed somehow reminds of a discrete wavelet transform, where the two indices $k$ and $n$ over which the double sum runs correspond to the scale and localization indices respectively (see, e.g. \cite{Akansu}, pp. 396-401) and the integrand is evaluated on the dyadic rationals.

Let us consider some mathematical examples of the application of the theorem. We can take $f(t)=\ln t$ in the interval $[x,1)$, where $x\in (0,1)$. On one hand, from Eq. (\ref{ectheo3}) we have
\begin{eqnarray}
\int_{x}^{1}\ln t \mathrm{d}t = \int_{1}^{x}\ln \frac{1}{t} \mathrm{d}t = -x\ln x+\sum_{k=-\lfloor \log_2 \ln (1/x) \rfloor}^{\infty}\frac{1}{2^{k}}\sum_{1}^{\left \lfloor 2^{k}\ln (1/x) \right \rfloor}(-1)^{n}\mathrm{e}^{-n/2^{k}}  \label{int3digitlog}
\end{eqnarray}
and, on the other hand, we have
\begin{equation}
\int_{1}^{x}\ln \frac{1}{t} \mathrm{d}t=-x\ln x+x-1 .
\end{equation}
Thence, any real number in the unit interval can be expanded as a sum of weighted exponentials,
\begin{equation}
x=1+\sum_{k=-\lfloor \log_2 \ln (1/x) \rfloor}^{\infty}\sum_{n=1}^{\left \lfloor 2^{k}\ln (1/x) \right \rfloor}\frac{(-1)^{n}}{2^{k}}\mathrm{e}^{-n/2^{k}} . \label{int3digitlog}
\end{equation}
The importance of Eq. (\ref{int3digitlog}) lies in the fact that introducing any power $2^s$ ($s \in \mathbb{Z}$) of $x$ simply shifts the scale; note that $x^{2^s}$ belongs to the unit interval if $x$ does.
Indeed, replacement of $x$ by $x^{2^s}$ transforms Eq. (\ref{int3digitlog}) into
\begin{eqnarray}
x^{2^s}&=&1+\sum_{k=-\lfloor \log_2 \ln (1/x) \rfloor+s}^{\infty}\sum_{n=1}^{\left \lfloor 2^{k-s}\ln (1/x) \right \rfloor}\frac{(-1)^{n}}{2^{k}}\mathrm{e}^{-n/2^{k}} \nonumber \\
&=&1+\sum_{k=-\lfloor \log_2 \ln (1/x) \rfloor}^{\infty}\sum_{n=1}^{\left \lfloor 2^{k}\ln (1/x) \right \rfloor}\frac{(-1)^{n}}{2^{k+s}}\mathrm{e}^{-n/2^{k+s}}.
\label{int3digitlogP}
\end{eqnarray}
Comparison of Eqs. (\ref{int3digitlog}) and (\ref{int3digitlogP}) shows that, while the limits in the sums are the same in both cases, the summands are all scale shifted so that one has the transformation $2^k \to 2^{k+s}$. This can be compared to the effect that multiplying by $2^s$ has on the binary radix expansion of $x$, (Eq. (\ref{idenreal}) with $p=2$),
\begin{eqnarray}
2^s x&=&\text{sgn}(2^s x)\sum_{k=-\infty}^{\lfloor \log_{2}|x| \rfloor+s} 2^{k} \mathbf{d}_{2}(k,2^s |x|) \nonumber \\
&=&\text{sgn}(x)\sum_{k=-\infty}^{\lfloor \log_{2}|x| \rfloor+s} 2^{k} \mathbf{d}_{2}(k-s, |x|) \nonumber \\
&=&\text{sgn}(x)\sum_{k=-\infty}^{\lfloor \log_{2}|x| \rfloor} 2^{k+s} \mathbf{d}_{2}(k,|x|) . \label{idenreal2}
\end{eqnarray}
Equation (\ref{int3digitlogP}) makes apparent that exponentiation of $x$ by $2^s$ also shifts the scale of $x$ and, hence, is a convenient representation to describe these shifts.

	Direct application of Eq. (\ref{ectheo}) of the theorem yields the following expansions for any non-negative real number $x$
\begin{eqnarray}
x&=&\int_{0}^{x}\mathrm{d}t=\sum_{k=-\lfloor \log_2 x \rfloor}^{\infty}\frac{1}{2^{k}}\sum_{n=1}^{\left \lfloor 2^{k}x \right \rfloor}(-1)^{n+1}=\sum_{k=-\lfloor \log_2 x \rfloor}^{\infty}\frac{1+(-1)^{1+\left \lfloor 2^{k}x \right \rfloor}}{2^{k+1}} \label{x} \\
x^2&=&2\int_{0}^{x}t\mathrm{d}t=\sum_{k=-\lfloor \log_2 x \rfloor}^{\infty}\frac{1}{2^{2k-1}}\sum_{n=1}^{\left \lfloor 2^{k}x \right \rfloor}(-1)^{n+1}n=\sum_{k=-\lfloor \log_2 x \rfloor}^{\infty}\frac{(-1)^{1+\left \lfloor 2^{k}x \right \rfloor}}{2^{2k-1}}\left \lfloor \frac{1+\left \lfloor 2^{k}x \right \rfloor}{2} \right \rfloor
\end{eqnarray}
The theorem can also be used to derive series and numerical algorithms to calculate certain transcendental numbers, or to find alternative expressions for them, such as
\begin{eqnarray}
\frac{\sqrt{\pi}}{2}&=&\int_{0}^{\infty}\mathrm{e}^{-x^{2}}\mathrm{d}x=
\sum_{k=-\infty}^{\infty}\frac{1}{2^k}\sum_{n=1}^{\infty}(-1)^{n+1}\mathrm{e}^{-n^2/2^{2k}} =
\sum_{k=1}^{\infty}\frac{1}{2^k}\sum_{n=1}^{2^{k}}(-1)^{n+1}\sqrt{\ln \frac{2^{k}}{n} } . 
\end{eqnarray}

An application of interest in number theory concerns the logarithmic integral that appears in the prime number theorem \cite{Edwards}.  From Eq. (\ref{ectheo}) we have, for $x>2$, $x\in \mathbb{R}$
\begin{equation}
\text{Li}(x):= \int_{2}^{x}\frac{\mathrm{d}t}{\ln t}=\sum_{k=-\lfloor \log_2 x \rfloor}^{\infty}\frac{1}{2^{k}}\sum_{n=\left \lfloor 2^{k+1} \right \rfloor+1}^{\left \lfloor 2^k x \right \rfloor}\frac{(-1)^{n+1}}{\ln n - k\ln 2}  \label{LieqA}
\end{equation}
Convergence of the series in Eq.  (\ref{LieqA}) is warranted and the integrals can be calculated to arbitrary precision by truncating the series to a sufficiently large value of $k$. The partial sum
\begin{equation}
\text{Li} (x; P) := \sum_{k=-\lfloor \log_2 x \rfloor}^{-\lfloor \log_2 x \rfloor+P}\frac{1}{2^{k}}\sum_{n=\left \lfloor 2^{k+1} \right \rfloor+1}^{\left \lfloor 2^k x \right \rfloor}\frac{(-1)^{n+1}}{\ln n - k\ln 2} \label{Plin}
\end{equation}
satisfies $\lim_{P\to\infty}\text{Li}(x;P)=\text{Li}(x)$. The representations of the functions $\text{Li}(x)$ and $\text{Li} (x; 10)$  are indistinguishable for the values of $x$ shown in Fig. \ref{Lipo}. 

\begin{figure}
\includegraphics[width=1.0\textwidth]{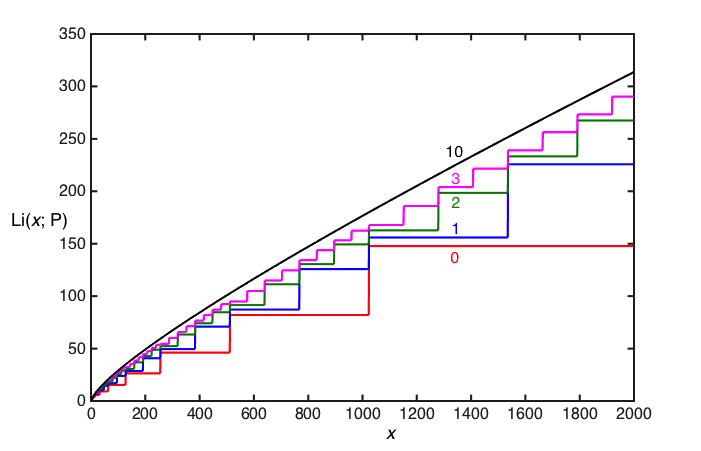}
\caption{\scriptsize{The function $\text{Li}(x; P)$ obtained from Eq. (\ref{Plin}) for the values of $P$ indicated close to the curves. For $P=10$ the curve is visually indistinguishable to the curve $P=\infty$ that yields the exact value of the integral.}} \label{Lipo}
\end{figure}

\section{Application in mathematical physics: Finite motion of classical Hamiltonian systems}\label{Hamiltonian}

Let us consider the classical motion of a particle with one degree of freedom and Hamiltonian
\begin{equation}
H=\frac{p^2}{2m}+U(q)
\end{equation}
where $m$ is the mass of the particle, $p$ its momentum, $q$ its position and $U(q)$ the potential energy. Since the total energy $E$ is conserved, the momentum can be obtained as
\begin{equation}
p=\pm \sqrt{2m[E-U(q)]}
\end{equation}
and, from Hamilton's equation of motion
\begin{equation}
\dot{q}=\dfrac{\partial H}{\partial p}=\pm \sqrt{\frac{2}{m}\left[E-U(q) \right]}
\end{equation} 
which can be integrated as \cite{Landau}
\begin{equation}
t=\int \frac{\mathrm{d}q}{\sqrt{\frac{2}{m}\left[E-U(q) \right]}}+\text{constant}.
\end{equation}

If the particle moves inside a potential well, then its motion is periodic. The above integral, evaluated between the turning points $q_1$ and $q_2$ at which the potential energy $U(q)$ equals the total energy $E$, gives the period   
\begin{equation}
T=\oint \frac{\mathrm{d}q}{\sqrt{\frac{2}{m}\left[E-U(q) \right]}}=\sqrt{2m}\int_{q_1}^{q_2} \frac{\mathrm{d}q}{\sqrt{E-U(q)}}.
\end{equation}
By using the theorem presented in this manuscript (corollary 1) we have
\begin{equation}
T=\sqrt{2m}\sum_{k=-\lfloor \log_2 (q_2-q_1) \rfloor}^{\infty}\sum_{n= 1}^{\left \lfloor 2^{k}(q_2-q_1) \right \rfloor}   \frac{(-1)^{n+1}}{2^{k}\sqrt{E-U\left(q_1+\frac{n}{2^{k}} \right)}}
\end{equation}
For example, for a nonlinear pendulum with potential energy $U(\theta)=-U_0\cos \theta$, with $\theta$ being the angle with respect to the vertical line and $U_{0}$ the potential energy minimum  \cite{Chernikov}, we have
\begin{equation}
T=\sqrt{\frac{2m}{E}}\sum_{k=-\lfloor \log_2 (2\theta_2) \rfloor}^{\infty}\sum_{n= 1}^{\left \lfloor 2^{k+1}\theta_2 \right \rfloor}   \frac{(-1)^{n+1}}{2^{k}\sqrt{1+\frac{U_0}{E} \cos \left(\theta_2-\frac{n}{2^{k}} \right)}}
\end{equation}
since, in a libration motion, the pendulum oscillates between the angles $q_1=\theta_1=-\theta_2$ and $q_2=\theta_2=\text{arccos}(-E/U_0)$. By using that $\cos \theta_{2}=1-2\sin^2(\theta_2/2)$, this can equivalently be written as
\begin{equation}
T=\sqrt{\frac{2m}{E+U_0}}\sum_{k=-\lfloor \log_2 (2\theta_2) \rfloor}^{\infty}\sum_{n= 1}^{\left \lfloor 2^{k+1}\theta_2 \right \rfloor}   \frac{(-1)^{n+1}}{2^{k}\sqrt{1-\eta^2 \sin^2 \left(\frac{\theta_2}{2}-\frac{n}{2^{k+1}} \right)}}
\end{equation}
where we have defined $\eta := \sqrt{2U_0/(E+U_0)}$. The integral diverges for $E=U_0$ (i.e. $\theta_2=\pi$) in which case the pendulum experiences a transition from libration to rotation. 

The motion of the pendulum is usually expressed with help of incomplete elliptic integrals of the first kind, defined by
\begin{equation}
F(\varphi | h)=\int_{0}^{\varphi}\frac{\mathrm{d}\theta}{\sqrt{1-h\sin^2\theta}}.
\end{equation}
The theorem indeed allows to numerically evaluate these integrals, truncating the series over $k$ to $P$ terms, as 
\begin{eqnarray}
F(\varphi | h, P)&=&\sum_{k=-\lfloor \log_2 \varphi \rfloor}^{-\lfloor \log_2 \varphi \rfloor+P}\sum_{n= 1}^{\left \lfloor 2^{k}\varphi \right \rfloor}   \frac{(-1)^{n+1}}{2^{k}\sqrt{1-h \sin^2 \left(\frac{n}{2^{k}} \right)}}. \label{FP}
\end{eqnarray}
For $P=10$, and the values of $\varphi=\pi/4$, $\pi/3$ and $\pi/2$ (Fig. \ref{figFP}), the functions $F(\varphi | h, P)$ are visually indistinguishable from the exact $F(\varphi | h)$ because the theorem implies
\begin{equation}
\lim_{P\to \infty}F(\varphi | h, P)=F(\varphi | h).
\end{equation}

\begin{figure}
\includegraphics[width=0.9\textwidth]{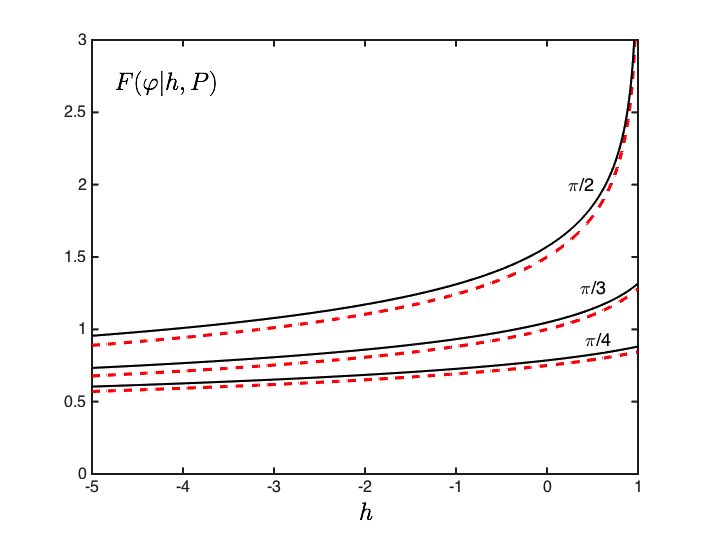}
\caption{\scriptsize{The partial sums $F(\varphi | h, P)$ given by Eq. (\ref{FP}) for the values of $\varphi$ indicated close to the curves and for $P=3$ (dashed curves) and $P=10$ (continuous curves).}} \label{figFP}
\end{figure}

Besides other methods, like the perturbation expansion of Kryloff and Bogoliuboff \cite{Fulcher}, the theorem presented here provides an alternative means to evaluate the elliptic integrals that govern the periodic motion of the simple pendulum. The generality of the theorem allows arbitrary potential functions to be handled in the same manner.

\section{Conclusions}

In this article, we have proved a theorem that allows to write an integral in terms of a convergent double series. The theorem has a wide applicability since the only requirement for the function in the integrand is being differentiable in the interval of integration.
The series can be truncated at any order to yield arbitrarily close approximations to the integral. 
 
We have illustrated with some mathematical examples the application of the theorem and we have provided an application to finding the period of finite motions of Hamiltonian systems in one dimension. Although the resulting integral can be solved by alternative methods, e.g. Gaussian numerical integration \cite{Percival}, quadrature rules only yield accurate approximations if the integrand is well approximated by polynomials of sufficiently low order. Our result is more general and does not require this.

\section*{Data availability statement}

Data sharing is not applicable to this article as no new data were created or analyzed in this study.
 

\end{document}